\documentclass[peerreview]{IEEEtran/IEEEtran}
\usepackage{amsmath,graphicx,amssymb}
\usepackage{mathtools,dsfont,amsfonts,units,amsthm,bm}
\usepackage{algorithm}
\usepackage{enumerate}

\usepackage{bbm}
\usepackage{tcolorbox}

\usepackage[noend]{algpseudocode}
\newtheorem{theorem}{Theorem}[section]
\newtheorem{definition}{Definition}[section]
\newtheorem{lemma}{Lemma}[section]
\newtheorem{corollary}{Corollary}[section]

\begin{document}
\title{Sensor Array Design Through Submodular Optimization}
\author{Gal~Shulkind,~\IEEEmembership{Student Member,~IEEE,} Stefanie Jegelka\thanks{This work was supported, in part, by NSF CAREER Award 1553284.},
	Gregory~W.~Wornell,~\IEEEmembership{Fellow,~IEEE}
\thanks{This work was supported, in part, by DARPA under Contract No. HR0011-16-C-0030, and by NSF under Grant No. CCF-1319828.}%
\thanks{The authors are with the Department of Electrical Engineering
	and Computer Science, Massachusetts Institute of Technology, Cambridge,
	MA 02139 USA (e-mail: \{shulkind,stefje,gww\}@mit.edu).}}

%
\maketitle
\begin{abstract}
We consider the problem of far-field sensing by means of a sensor array. Traditional array geometry design techniques are agnostic to prior information about the far-field scene. However, in many applications such priors are available and may be utilized to design more efficient array topologies.

We formulate the problem of array geometry design with scene prior as one of finding a sampling configuration that enables efficient inference, which turns out to be a combinatorial optimization problem. While generic combinatorial optimization problems are NP-hard and resist efficient solvers, 
we show how for array design problems the theory of submodular optimization may be utilized to obtain efficient algorithms that are guaranteed to achieve solutions within a constant approximation factor from the optimum.
We leverage the connection between array design problems and submodular optimization and port several results of interest. We demonstrate efficient methods for designing arrays with constraints on the sensing aperture, as well as arrays respecting combinatorial placement constraints.

This novel connection between array design and submodularity suggests the possibility for utilizing other insights and techniques from the growing body of literature on submodular optimization in the field of array design.
\end{abstract}

\section{Introduction} \label{sec:intro}
Sensor arrays for spatial sensing are widely deployed in a wide range of applications including radar, sonar, medical imaging and radio astronomy and there is a vast literature on the topics of array design and array processing from the last century \cite{van2004detection,krim1996two}.

A major goal in designing arrays is efficiently meeting sensing specifications with a limited budget of sensing elements, which is often a main determinant of system cost in terms of dimensions, weight, complexity and manufacturing cost. However, traditionally most studies assume a uniform and linear design, often a truncated half-wavelength array, or restrictions thereof. Indeed, the problem of designing non-uniform arrays hints at combinatorial optimization and is computationally intractable as we discuss later.

In beamforming arrays attaining a desired resolution level is often a primary concern, achieved by means of manipulating beam pattern parameters such as lobe widths and positions. The problem of designing efficient array geometries that facilitate desired beam patterns has been widely studied in the past. Techniques involving array thinning start with a dense uniform geometry, removing elements while maintaining performance within specified bounds \cite{haupt1994thinned}. Other approaches consider various methods such as swarm optimization \cite{khodier2005linear}, dynamic programming \cite{skolnik1964dynamic}, genetic algorithms \cite{haupt2007genetic}, inversion \cite{kumar1999design} and Bayesian compressive sampling \cite{oliveri2012complex}.

In direction-of-arrival estimation applications other specialized techniques have emerged for finding efficient array designs, such as optimizing the corresponding Cramer-Rao error bound \cite{gazzah2006cramer}, or designing according to the nested array methodology for increasing the available number of degrees of freedom \cite{vaidyanathan2011sparse}, and theoretical results for estimation performance in these settings have been developed \cite{stoica1990performance}, \cite{wang2017coarrays}, \cite{koochakzadeh2016cramer}. Alternatively, other studies have addressed array design from the point of view of performing efficient compressive sampling of a Wide Sense Stationary (WSS) process, e.g. for a setting with a parametrized covariance matrix \cite{romero2015compression}, for wide-band power spectrum estimation \cite{ariananda2012compressive} and evaluating performance in the presence of modeling errors \cite{koochakzadeh2017performance}. However, these samplers are optimized for reconstructing the second order statistics of the scene and not the actual scene realization.

In virtually all design methodologies some assumptions are made on the scene of interest. These represent beliefs, constraints or knowledge that hold over the unobserved scene. A favorable design is one meeting requirements while taking into account these assumptions. For example, in direction of arrival applications we may assume some limit on the number of point targets present in the scene \cite{gazzah2006cramer}. In beamforming we assume some separation level between objects of interest that implies a certain resolution level is required \cite{van2004detection}, or some scene sparsity structure \cite{krieger2013design,krieger2014multi}.

In this paper we study the problem of inference on a scene of interest through measurements collected by an array of sensors. The approach we take for designing array geometries is novel in that we consider settings where some Bayesian prior on the scene is available at the time of design and propose exploiting this knowledge by adapting the array geometry accordingly to achieve efficient inference with a limited budget for sensors.

Prior knowledge is in fact not a rarity at all. Frequently, the same device is used to sense multiple similar scenes, where past examples are indicative. A medical imaging device, for instance, is typically used to image the same organ across different patients.
In other cases, we may have prior knowledge in the form of scene properties such as smoothness or adherence to spatial constraints. We incorporate such knowledge as a prior in a Bayesian model, such that sensing the scene is equivalent to performing inference in this model, and the problem of array geometry design asks to select a geometry that optimizes the quality of inference. This naturally turns out to be a combinatorial optimization problem, which can be extremely hard to solve even approximately without exploiting additional structure.

We show how to choose a cost function for the array geometry design problem that holds the property of submodularity \cite{nemhauser1978analysis}. A submodular set function is one that exhibits diminishing marginal gains, i.e. adding additional elements results in diminishing benefit. Recently, there has been significant progress on the theory of optimizing submodular functions \cite{fujishige2005submodular,buchbinder2015tight,vondrak2013symmetry,mirzasoleiman2013distributed}. In particular, these results state that, while NP-hard, submodular maximization admits variants of greedy algorithms that are guaranteed to achieve near-optimal solutions, i.e. within a constant factor. Submodularity has been used in connection with sensor placement problems, for example in \cite{krause2008near,shamaiah2010greedy,krause2011submodularity}, however these works are not tailored to the far field scenes and models that we focus on here.

Our connections and formulations open avenues for leveraging those results for efficient array design with strong guarantees. 
We demonstrate this by showcasing the design of array geometries in settings with arbitrary apertures and settings where sensor placement respects matroid combinatorial constraints by porting results from the growing body of literature on submodular optimization. We demonstrate that together with  our new adaptive formulations, the exploitation of prior knowledge leads to higher quality inference at lower cost in terms of the number of sensing elements.

The paper is structured as follows: In Section \ref{sec:problem_formulation}, we review far-field sensing and introduce a statistical framework for modeling a-priori scene distributions. In Section \ref{sec:array_design}, we define a cost-function to quantify the inference loss associated with array configurations and derive a corresponding array design optimization problem. We show that this optimization problem belongs to the class of submodular optimization problems and adapt established algorithms and guarantees to our setting. In Section \ref{sec:matroid_constraints}, we review the theory of matroids and showcase its use in designing array geometries with combinatorial placement constraints. Finally, in Section \ref{sec:numerical_experiments} we summarize numerical experiments exemplifying and validating our theory. The appendix provides proofs for some of our claims.
\section{Problem formulation} \label{sec:problem_formulation}

In this section we formulate the array design problem. We focus on far-field sensing applications, although, as will become apparent, the same techniques could also be generalized to other settings where the measurement process is linear. For simplicity we consider a one dimensional setting as the extension to multiple dimensions is straightforward. We begin with a review of the far-field sensing model to establish notation, and then pose the sensing problem as a Bayesian inference problem.
\subsection{Far-Field Sensing} \label{sec:far_field_imaging} 
The far-field sensing setup is depicted in Fig.~\ref{fig:array}. At a distance from an observation axis $x$ we have a scene of interest, characterized by an illumination function $\tilde{\beta}(\theta)$ that depends on the angle $\theta{\in}[-\frac{\pi}{2},+\frac{\pi}{2}]$ between the direction of observation and the axis of measurement broadside. We are to place $N$ sensors along the observation axis at positions $\mathcal S\equiv\left\lbrace x_1,\ldots,x_{N} \right\rbrace $. 

\begin{figure}[h] 
	\begin{center}
		\includegraphics[scale=1]{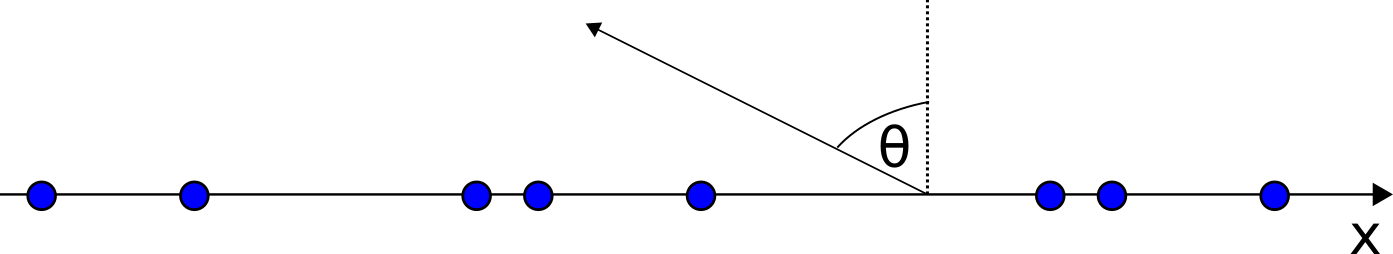}
	\end{center}
	\caption{Far-field sensing. Antenna elements are depicted as blue dots.} \label{fig:array}
\end{figure}

The passive sensors pick up the radiation emitted from the scene, which is assumed to be narrow-band around a fixed wavelength $\lambda$.
Recall that in an ideal noiseless setting, $r(x)$, the reading collected at position $x$, is given according to \cite{van2004detection}:
\begin{align} \label{eq:array_imaging_eq}
r(x){=}\int\limits_{-\frac{\pi}{2}}^{+\frac{\pi}{2}}\tilde{\beta}(\theta)e^{j\frac{2\pi}{\lambda} x\sin(\theta)}\cos(\theta)d\theta{=}\int\limits_{-\frac12}^{+\frac12}\beta(\psi)e^{j \frac{4\pi}{\lambda} x\psi}d\psi,
\end{align}
where we have introduced $\psi\equiv\frac{\sin\theta}{2}$ such that $\psi\in[-\frac12,+\frac12]$ and $\beta(\psi)\equiv\tilde{\beta}({\sin^{-1}(2\psi)})$.

Eq. \eqref{eq:array_imaging_eq} bears close resemblance to the definition of the Fourier transform. Namely, defining $f(t)$ to be the Fourier transform pair of $\beta(\psi)$
\begin{align} \label{eq:Fourier_pair_definition}
f(t)=\int\beta(\psi)e^{j2\pi t \psi}d\psi \; \mathrel{\mathop{\rightleftarrows}^{\mathrm{\mathcal{F}}}_{\mathrm{\mathcal{F}^{-1}}}} \;\beta(\psi){=}\int f(t)e^{-j2\pi t \psi}d t
\end{align}
we immediately identify, comparing \eqref{eq:array_imaging_eq} and \eqref{eq:Fourier_pair_definition} that $r(x)=f(\frac{2}{\lambda} x)$. That is, collecting measurements at $x\in\left\lbrace x_n \right\rbrace$  is equivalent to sampling the function $f(t)$ at $t\in\left\lbrace \frac{2}{\lambda}x_n \right\rbrace$. Notice that since the support of $\beta(\psi)$ is restricted to $\psi\in [-\frac12,+\frac12]$ we observe that $f(t)$ is band-limited with bandwidth $1$. Thus, as a consequence of the Whittaker-Kotelnikov-Shannon sampling theorem we have that perfect reconstruction of $f(t)$, and subsequently the scene $\beta(\psi)$, is possible from an infinite set of samples taken at $t\in\left\lbrace \ldots,-1,0,+1,\ldots\right\rbrace$ which are exactly the samples collected by an infinite $\frac{\lambda}{2}$-spaced array.

Finally, we consider the effect of noise by introducing the vector of $N$ noisy measurements $\tilde{\boldsymbol{f}}=[\tilde{f}_{1},\ldots,\tilde{f}_{N}]^{\top}$ modelled according to:
\begin{align} \label{eq:collecting_noisy_samples}
\tilde{f}_{n}\equiv r(x_n)+w_{n}=f\left(\frac{2}{\lambda}x_n\right)+w_{n} & & n=1,\ldots,N
\end{align}
where $w_{n}$ are additive noise components. Stacked into a vector $\boldsymbol{w}=[w_1,\ldots,w_{N}]^{\top}$, we assume throughout that the noise is a complex, circular, Gaussian random vector $\boldsymbol{w} \sim\mathcal{CN}(0,\boldsymbol{\Sigma}_{ww})$ \cite{neeser1993proper}. 
We assume the noise to be independent across different sensors, i.e., $\boldsymbol{\Sigma}_{ww}=\sigma^2_{w}\boldsymbol{I}_{N}$, where $\boldsymbol{I}_{N}$ is an $N{\times} N$ identity matrix such that the noise is i.i.d. across different sensors.
\subsection{Bayesian Formulation} \label{sec:statistical_description}
The sensing problem we consider in this paper entails the estimation of the illumination function $\beta(\psi)$ or equivalently its inverse Fourier transform $f(t)$ from the set of noisy measurements $\tilde{\boldsymbol{f}}$. Even in the noiseless setting this problem is gravely ill-posed as infinitely many wildly varying scenes map to any given finite set of observations\footnote{The mapping between a band-limited function $f(t)$ and a finite set of its (generally non-uniform) samples $\left\lbrace f(t_n)\right\rbrace$ is not bijective \cite{yen1956nonuniform}.}.

To cope with this ill-posedness, some prior belief or knowledge pertaining to $\beta(\psi)$ (or equivalently $f(t)$) must hence be incorporated into the model, and this could be achieved in several ways. Wingham \cite{wingham1992reconstruction} proposed selecting one specific $f(t)$ of the multiple such functions consistent with the samples, namely the minimum norm solution. Alternatively, some constraints or other preferences may be imposed on the solution by penalizing the inversion. For example one may require the solution to satisfy some constraints, e.g. lie in some pre-specified sub-space of the function space, or regularize the inversion, e.g. to promote smoothness or low total variation \cite{eldar2015sampling}.

In this study we take a Bayesian approach and impose a prior distribution on the scene $\beta(\psi)$. Subsequently, sensing is equivalent to performing inference in this model. The prior may be assigned based on past observations over the distribution of scenes or based on a-priori knowledge of scene properties.

We conceptually think of our scene as a natural image, and as such, our Bayesian formulation entails posing a prior that is consistent with its expected characteristics, tailored to the specific application that is of interest to the user. Utilizing Gaussian distributions as priors for natural scenes and images has been successfully applied in the past \cite{mihcak1999low}, \cite{levin2007image}, \cite{levin2009understanding}  and our work uses similar notions, with the Gaussian prior specifically expressed in the frequency domain (e.g., as in [9]). We note that this Gaussian description naturally captures the statistics of spread-out scenes and is not immediately suitable for serving as a prior for point-source target models. 
The formulations and methods that we develop in subsequent sections can be adapted for use with other scene priors in lieu of the Gaussian one we focus on for this study. In particular, our formulations may account for settings with point-source scenes by posing a point-process prior (e.g. a determinantal point process, DPP \cite{kulesza2012determinantal}) that is better suited for capturing priors of point-targets in specific environments. With this alternative prior we could follow similar steps to develop efficient array designs for those applications. Studying such alternative priors and resulting formulations is a possible interesting extension to our work.

In what follows, to be concrete we detail one specific choice for the form of the scene prior, expressed as a Gaussian distribution in the frequency domain:
\subsubsection{Discrete Representation}
Assigning a prior on a continuous function $\beta(\psi)$ may seem like a daunting task. However, in many scenarios this can be considerably simplified if the scene may be efficiently expanded in a countable basis of functions such that the prior may be imposed in the discrete domain of expansion coefficients. In what follows, we consider expanding $\beta(\psi)$ by means of Fourier basis functions $\left\lbrace e^{j 2\pi m\psi}\vert m=\ldots,-1,0,+1,\ldots\right\rbrace $ in $\psi\in[-\frac12,+\frac12]$. Alternatively, other expansions may be used with similar results. We write:
\begin{align} \label{eq:beta_f_expansion}
\beta(\psi)&=\sum\limits_{m} \beta_m e^{j 2\pi m\psi}, \qquad \psi\in[-\frac12,+\frac12]\\
\beta_m&\equiv\int\limits_{{-}\frac12}^{{+}\frac12} \beta(\psi) e^{-j2\pi m\psi} d \psi
\end{align}
where $\left\lbrace \beta_m\right\rbrace$ are the Fourier expansion coefficients. The usual Parseval relation holds:
\begin{align} \label{eq:parseval}
\int \vert\beta(\psi)\vert^2 d\psi = \sum\limits_{m} \vert \beta_m\vert^2
\end{align}
In lieu of the prior on $\beta(\psi)$ we will impose a prior distribution over $\left\lbrace \beta_m \right\rbrace $. As will become evident such a description is especially suited for applications involving real world smooth scenes $\beta(\psi)$ as suggested by the following properties of the Fourier series expansion \cite{Fourierprop}:
\begin{enumerate}[(a)]
	\item (Riemann-Lebesgue lemma) Let $\beta(\psi)$ be any integrable function. Then $\vert\beta_m\vert\overset{\vert m \vert\rightarrow \infty}{\longrightarrow}0$.
	\item Let $\beta(\psi)\in C^r$ where $C^r$ is the space of $r$-times continuously differentiable functions over some domain. Then $\vert\beta_m\vert\leq\frac{\alpha}{\vert m\vert^r}$ with $\alpha=\sup_\psi \vert \frac{\partial^r}{\partial\psi^r}\beta(\psi)\vert$ \cite{Fourierprop}.
\end{enumerate}
Thus, for any nicely behaved scene $\beta(\psi)$ the high frequency Fourier expansion coefficients diminish to zero with an asymptotically polynomial rate determined by the level of smoothness.
\subsubsection{Prior} \label{secsub:prior}
In the sequel we use Gaussian priors on the coefficients $\left\lbrace \beta_m\right\rbrace $:
\begin{align} \label{eq:prior_gaussian}
\beta_m \sim \mathcal{CN}(0,\sigma_m^2)
\end{align}
where $\beta_m$ are independent, complex, circular and Gaussian, and $\sigma_m^2$ are the corresponding variances, which are assumed known in advance. 
Using \eqref{eq:parseval} we define the expected scene power:
\begin{align}
P\equiv\mathbb{E}\int \vert\beta(\psi)\vert^2 d\psi=\sum\limits_{m} \sigma_m^2
\end{align}
The $\left\lbrace \sigma_m^2\right\rbrace $ can be set following some initial measurements of sample functions $\beta(\psi)$ or taking into account prior knowledge. For example if we have a-priori knowledge that $\beta(\psi)\in C^r$ we may use
\begin{align} \label{eq:prior_def}
\sigma_m^2=\left\lbrace\begin{array}{ll}
1 & m=0 \\
\frac{1}{m^{2r}} & m\neq 0
\end{array}\right.
\end{align}
which is a very simple distribution respecting the polynomial variance decay.
For the rest of the paper we adopt the prior in \eqref{eq:prior_def} and take $r=1$  to promote continuously differentiable functions.
Also notice that the resulting prior distribution on $\beta(\psi)$ is Wide Sense Stationary (WSS), with the coefficients $\sigma_m^2$ determining the shape of the power spectrum, based on the available prior information we have on the scene.
\subsection{Observation Model}
With the prior stated in the discrete $\left\lbrace \beta_m\right\rbrace $ domain as described above, our next goal is to circumvent $\beta(\psi)$ by directly stating the problem in terms of the measurements $\tilde{f}_n$ and the coefficients $\beta_m$, replacing the continuous representation with manageable discrete counterparts. Substituting the sum \eqref{eq:beta_f_expansion} into Equation \eqref{eq:array_imaging_eq} we have:
\begin{align} \label{eq:orig_model}
r(x_n)=\int\limits_{-\frac12}^{+\frac12}\sum\limits_{m} \beta_m e^{j 2\pi m\psi}e^{j \frac{4\pi}{\lambda}x_n\psi}d\psi =\sum\limits_{m}K_{nm} \beta_m 
\end{align}
where 
\begin{align} \label{eq:K_definition}
K_{nm}\equiv\int\limits_{-\frac12}^{+\frac12}e^{j2\pi (\frac{2}{\lambda}x_n+m) \psi } d \psi= \text{sinc}(m{+}\frac{2}{\lambda}x_n)
\end{align}
with $\text{sinc}(x)\equiv \frac{\sin(\pi x)}{\pi x}$. Plugging this into \eqref{eq:collecting_noisy_samples} we finally have the full observation model
\begin{align} \label{eq:full_model}
\tilde{f}_{n}=\sum\limits_{m}K_{nm} \beta_m+w_{n} & & \begin{array}{l}
n=0,\ldots,N-1
\end{array}
\end{align}
and the sensing problem amounts to estimating the coefficients $\left\lbrace \beta_m\right\rbrace $ given the observation vector $\tilde{\boldsymbol{f}}$. As we have assumed a Gaussian distribution for the noise vector $\boldsymbol{w}$ as well as for the prior over $\left\lbrace \beta_m\right\rbrace $ the posterior distribution $p(\left\lbrace \beta_m\right\rbrace \vert \tilde{\boldsymbol{f}})$ turns out to be Gaussian with a convenient analytic expression, as we detail in Section \ref{sec:reconstruction}.
 
\section{Array Design With an Aperture Constraint} \label{sec:array_design}
In the previous section we formulated the problem of far-field sensing in a Bayesian setting with a prior on the distribution of the underlying scene. In this section we design corresponding array geometries to facilitate efficient sensing that exploit the prior knowledge and model. Initially we consider simple constraints on the total number of sensors we can place and on the aperture where these can be situated. Specifically, we assume that up to $N$ sensors are free to be placed over some aperture $\mathcal A$ on the real line, e.g. for simplicity we can consider $\mathcal A\equiv\left\lbrace x\vert -a\leq x \leq a\right\rbrace,  a\in \mathbb{R}^+ $. In Section \ref{sec:matroid_constraints} we consider more sophisticated combinatorial placement constraints and show that the same formulations we develop here may be adapted to those more challenging use cases.
\subsection{Formulating a Cost Function}
In order to make the problem of optimal array design well-posed in our Bayesian setting we need to specify a performance measure that will be used as a cost function to compare different designs and choose the best one. Revisiting our observation model \eqref{eq:full_model} we notice that the array design determines the coefficients $K_{mn}$ through the set of sensor positions $\mathcal S$, as defined in Section \ref{sec:far_field_imaging}. With any given design the sensing problem amounts to inferring the posterior $p(\left\lbrace \beta_m\right\rbrace \vert \tilde{\boldsymbol{f}})$. A natural performance measure in this setting quantifies the quality of inference, i.e., the information gained by performing the sensing experiments which results in updating our beliefs about the coefficients $\left\lbrace \beta_m\right\rbrace$ from the prior $p(\left\lbrace \beta_m\right\rbrace)$ to the posterior $p(\left\lbrace \beta_m\right\rbrace \vert \tilde{\boldsymbol{f}})$. This problem has been extensively studied in the context of statistical inference and experimental design \cite{chaloner1995bayesian},  \cite{bernardo1979expected}. If we are interested in making general inference to learn the state of the world (represented by the posterior distribution of the Fourier coefficients) then the natural performance measure, which is referred to as Bayes D-optimality is the mutual information between the expansion coefficients and the collected measurements \cite{chaloner1995bayesian}, \cite{bernardo1979expected}:
\begin{align}
I(\tilde{\boldsymbol{f}}_{\mathcal S} ; \left\lbrace \beta_m\right\rbrace )=H(\left\lbrace \beta_m\right\rbrace)-H(\left\lbrace \beta_m\right\rbrace\vert \tilde{\boldsymbol{f}}_{\mathcal S})
\end{align}
where $I(\cdot;\cdot)$ is the mutual information and $H(\cdot)$ the Shannon entropy. The subscript $\mathcal{S}$ explicitly emphasizes the dependence of the measurements on the set of sensor positions $\mathcal{S}$.
Notice that this cost function measures the reduction in uncertainty of the scene expansion coefficients before and after measurements are made. The larger the mutual information the more we trust the values of the coefficients $\left\lbrace \beta_m\right\rbrace\vert \tilde{\boldsymbol{f}}_{\mathcal S}$.

With the cost function in place the array design problem becomes
\begin{align} \label{eq:optimization_problem}
\mathcal{S}^{\star}=\underset{\mathcal S\subseteq \mathcal A, \vert \mathcal{S}\vert \leq N }{\text{argmax}}\, I(\tilde{\boldsymbol{f}}_{\mathcal S}  ; \left\lbrace \beta_m\right\rbrace),
\end{align}
which is an NP-hard combinatorial optimization problem. However, we will show later that we can obtain a constant-factor approximation for \eqref{eq:optimization_problem} using efficient computational techniques.
\subsection{Finite Dimensional Approximation}
Solving \eqref{eq:optimization_problem} under our model \eqref{eq:full_model} involves manipulations of the infinite sequence $\left\lbrace \beta_m\right\rbrace $ which may not be amenable to computer representation. To make our formulation tractable we approximate the infinite sequence $\left\lbrace \beta_m\right\rbrace$ with a finite truncated set of coefficients $\left\lbrace \beta_m  \vert m\in \mathcal{M} \right\rbrace$, where $\mathcal{M}$ is some finite set. Stacked in vector form $\boldsymbol{\beta}$, we consider the simplified finite dimensional approximation of \eqref{eq:full_model}:
\begin{align} \label{eq:appx_model}
\hat{\boldsymbol{f}}_{\mathcal S}=\boldsymbol{K}_{\mathcal S}\boldsymbol{\beta}+\boldsymbol{w}
\end{align}
where $\boldsymbol{K}_{\mathcal S}$ is an $N{\times}\vert\mathcal{M}\vert$ matrix formed by restricting $K_{nm}$ on $m\in\mathcal{M}$ and the dependency on the sampling set $\mathcal{S}$ again made explicit. Notice that hat notation is replacing the previous tilde.

We show that the approximate finite-dimensional model \eqref{eq:appx_model} is a good proxy for the original infinite-dimensional model \eqref{eq:full_model} for a suitably selected $\mathcal{M}$. Indeed, if $\mathcal{M}$ is chosen to only exclude those $\beta_m$ coefficients that in expectation contribute a marginally small part of the energy of the full infinite sequence $\left\lbrace \beta_m\right\rbrace$ then the mutual information derived from the approximate model \eqref{eq:appx_model} will closely track that of the infinite dimensional model in \eqref{eq:full_model}. More precisely we have the following result:
\begin{lemma} \label{lemma:1}
	Let the prior on $\left\lbrace \beta_m\right\rbrace $ be i.i.d. according to $\beta_m \sim \mathcal{CN}(0,\sigma_m^2)$, $\mathcal{M}$ fixed, and $\epsilon\equiv\sum\limits_{m\notin \mathcal{M}} \sigma^2_m$ satisfies $\epsilon<\sigma_w^2N^{-\frac32}$. We have: $${-} N\log(1{+}\frac{\epsilon  N^{\frac32}}{\sigma_{w}^{2}}) {\leq} I(\tilde{\boldsymbol{f}_{\mathcal S}}  ; \left\lbrace \beta_m\right\rbrace ){-}I(\hat{\boldsymbol{f}_{\mathcal S}}  ; \boldsymbol{\beta} ) {\leq} {-} N\log(1{-}\frac{\epsilon  N^{\frac32}}{\sigma_{w}^{2}})$$
\end{lemma} 
\begin{proof}
	See Appendix \ref{app:proof_of_truncation_lemma}.
\end{proof}
By virtue of the last lemma and $N\log(1\pm\frac{\epsilon  N^{\frac32}}{\sigma_{w}^{2}}) \overset{\epsilon\rightarrow 0 }{\longrightarrow} 0$ we have that $I(\hat{\boldsymbol{f}_{\mathcal S}}  ; \boldsymbol{\beta} )$ is an arbitrarily accurate proxy for $I(\tilde{\boldsymbol{f}_{\mathcal S}}  ; \left\lbrace \beta_m\right\rbrace )$ for $\epsilon$ small enough, such that in lieu of problem \eqref{eq:optimization_problem} we now continue with the simplified finite dimensional approximation:
\begin{align} \label{eq:simplified_opt}
\mathcal{S}^{\star}=\underset{ \mathcal{S}\subseteq \mathcal A, \vert \mathcal{S}\vert \leq N }{\text{argmax}}\, I(\hat{\boldsymbol{f}}_{\mathcal S}  ; \boldsymbol{\beta})
\end{align}
and the results will be accurate to within the approximation bounds from Lemma \ref{lemma:1}.
\subsection{Grid Discretization}
Next we turn to discretizing the aperture $\mathcal A$ to cast the array design problem in the form of a generic finite selection problem. We thus restrict the choice of sampling positions to the finite set
\begin{align}
\mathcal{V}\equiv\left\lbrace x^{1},\ldots,x^{\vert \mathcal{V} \vert}\right\rbrace \subset \mathcal A
\end{align} 
For the sequel we take $\mathcal{V}$ to be a uniform $\delta$-spaced grid of positions in $\mathcal A$. We adapt the array design problem \eqref{eq:simplified_opt} accordingly as:
\begin{align} \label{eq:trunc_disc_opt}
	\mathcal{S}^{\star}_d=\underset{\mathcal S\subseteq \mathcal{V}, \vert \mathcal S\vert \leq N }{\text{argmax}}\, I(\hat{\boldsymbol{f}}_{\mathcal S}  ; \boldsymbol{\beta})
\end{align}
with the subscript $d$ implying discretization. The next result can be used to quantify the  level of discretization $\delta$ necessary to guarantee performance close to optimum within some specified error bound:
\begin{lemma} \label{lemma:2}
	With $\mathcal{V}$ a uniform grid of sampling positions with distance $\delta$ between adjacent positions we have:
  \begin{align*}
  I(\hat{\boldsymbol{f}}_{\mathcal{S}_d^{\star}}  ; \boldsymbol{\beta}) \leq I(\hat{\boldsymbol{f}}_{\mathcal{S}^{\star}}  ; \boldsymbol{\beta})\leq I(\hat{\boldsymbol{f}}_{\mathcal{S}_d^{\star}}  ; \boldsymbol{\beta}) {+} N\log(1{+}\frac{4\delta P (1{+}\delta) N^{\frac32}}{\lambda\sigma_{w}^{2}})
  \end{align*}
\end{lemma} 
\begin{proof}
	See Appendix \ref{app:proof_of_discretization_lemma}.
\end{proof}
With this last lemma in place the array design problem \eqref{eq:simplified_opt} may be further approximated in the more convenient finite combinatorial problem form of \eqref{eq:trunc_disc_opt} with guarantees on the accuracy of the resulting designs. In the sequel we assume that $\delta$ is chosen such as to meet desired accuracy levels, as prescribed in Lemma \ref{lemma:2}, and work with the simplified formulation \eqref{eq:trunc_disc_opt}.
\subsection{Sensing} \label{sec:reconstruction}
With the observation model of \eqref{eq:appx_model}, coupled with a Gaussian distribution for the noise vector $\boldsymbol{w}$ and a Gaussian prior for the coefficients vector $\boldsymbol{\beta}$, calculation of the posterior distribution $\boldsymbol \beta\vert \hat{\boldsymbol f}_{\mathcal S}$ is particularly simple and can be performed analytically.
Concretely, we have as a result of all random variables being Gaussian $\boldsymbol \beta\vert \hat{\boldsymbol f}_{\mathcal S} \sim \mathcal{CN}(\hat{\boldsymbol\mu},\hat{\boldsymbol\Sigma})$ and the parameters are given according to the conventional Gaussian conditional parameters:
\begin{align}
\hat{\boldsymbol\mu}&= \boldsymbol \Sigma_{\beta\hat{f}}\boldsymbol\Sigma^{-1}_{\hat f \hat f}\hat{\boldsymbol f}_{\mathcal S}\notag\\
\hat{\boldsymbol\Sigma}&=\boldsymbol \Sigma_{\beta\beta}-\boldsymbol \Sigma_{\beta\hat{f}}\boldsymbol\Sigma^{-1}_{\hat f \hat f} \boldsymbol\Sigma_{\beta\hat{f}}^{\dagger}
\end{align}
where $\boldsymbol \Sigma_{\beta\hat{f}}=\boldsymbol\Sigma_{\beta\beta}\boldsymbol {K}_S^{\dagger},\quad \boldsymbol\Sigma_{\hat f \hat f}=\boldsymbol {K}_{\mathcal S}\boldsymbol\Sigma_{\beta\beta} \boldsymbol {K}_{\mathcal S} ^{\dagger}+\boldsymbol\Sigma_{ww}$ and $\boldsymbol\Sigma_{\beta\beta}=\text{diag}[\sigma_1^2,\ldots,\sigma_{\mathcal{M}}^2]$.
\subsection{Submodular Optimization}
The next step is to prescribe an efficient algorithm for the solution of \eqref{eq:trunc_disc_opt}. As we will show shortly \eqref{eq:trunc_disc_opt} is an instance of a known NP-hard problem such that it is widely believed that no efficient algorithm for its solution exists. However, due to the structure of the cost function an efficient approximation algorithm is known to exist with strong theoretical guarantees. In this subsection we survey the relevant results and adapt them to our needs.
\subsubsection{Submodularity}
We begin by invoking the submodularity property of set functions \cite{nemhauser1978analysis}:
\begin{definition}
	Let $G: 2^{\mathcal{V}}\rightarrow \mathbb{R}$ be a set function.\\
	(a) G  is \textbf{submodular} if it satisfies the property of decreasing marginals: $\; \forall \mathcal{S},\mathcal{T} \subseteq \mathcal{V}$ such that $\mathcal{S} \subseteq T$ and $x\in \mathcal{V}\backslash T$ it holds that $G(\mathcal{S}\cup\left\lbrace x \right\rbrace )-G(\mathcal S) \geq G(\mathcal T\cup\left\lbrace x \right\rbrace )-G(\mathcal T)$.\\
	(b) $G$ is \textbf{monotonic} (increasing) if $\; \forall \mathcal{S},\mathcal{T} \subseteq \mathcal{V}$ s.t. $\mathcal{S}\subseteq \mathcal{T}$ we have $G(\mathcal{S})\leq G(\mathcal{T})$.
\end{definition}
As it turns out, our cost function is monotonic and submodular as the next result shows (this is almost identical to corollary 4 in \cite{krause2012near}. We reproduce it here with adaptations to our setting for completeness):
\begin{lemma}
	Let $\mathcal{V}$ be defined as before, and define the set function $G: 2^{\mathcal{V}}\rightarrow \mathbb{R}$ according to $G(\mathcal S)=I(\hat{\boldsymbol{f}}_{\mathcal S}  ; \boldsymbol{\beta})$. Then $G$ is submodular and monotonic (increasing).
\end{lemma}
\begin{proof}
Expanding the mutual information according to $I(x;y)=H(x){-}H(x\vert y)$ we have:
\begin{align}
&G(\mathcal{S}{\cup}\left\lbrace x \right\rbrace ){-}G(\mathcal{S})=H(\hat{\boldsymbol{f}}_{\mathcal{S}\cup\left\lbrace x \right\rbrace}){-}H(\hat{\boldsymbol{f}}_{\mathcal S})\notag\\
&{-}[H(\hat{\boldsymbol{f}}_{\mathcal S\cup\left\lbrace x \right\rbrace}\vert \boldsymbol{\beta} ){-}H(\hat{\boldsymbol{f}}_{\mathcal S}\vert \boldsymbol{\beta})]=H(\hat{\boldsymbol{f}}_{\left\lbrace x \right\rbrace}\vert \hat{\boldsymbol{f}}_{\mathcal S}){-}H(\hat{\boldsymbol{f}}_{\left\lbrace x \right\rbrace}\vert \boldsymbol{\beta})
\end{align}
where in the last equality we used the conditional independence of the components of $\hat{\boldsymbol{f}}_{\mathcal S\cup\left\lbrace x \right\rbrace}$ given $\boldsymbol{\beta}$. Substituting $T$ for $\mathcal S$ we immediately get:
\begin{align}
&[G(\mathcal S{\cup}\left\lbrace x \right\rbrace ){-}G(\mathcal S)]{-}[G(\mathcal T\cup\left\lbrace x \right\rbrace ){-}G(\mathcal T)]\notag\\
&=H(\hat{\boldsymbol{f}}_{\left\lbrace x \right\rbrace}\vert \hat{\boldsymbol{f}}_{\mathcal S}){-}H(\hat{\boldsymbol{f}}_{\left\lbrace x \right\rbrace}\vert \hat{\boldsymbol{f}}_{\mathcal T})
\end{align}
Using $\mathcal S\subseteq \mathcal T$ we have $H(\hat{\boldsymbol{f}}_{\left\lbrace x \right\rbrace}\vert \hat{\boldsymbol{f}}_{\mathcal S})\geq H(\hat{\boldsymbol{f}}_{\left\lbrace x \right\rbrace}\vert \hat{\boldsymbol{f}}_{\mathcal T})$
such that $G(\mathcal S{\cup}\left\lbrace x \right\rbrace ){-}G(\mathcal S)\geq G(T{\cup}\left\lbrace x \right\rbrace ){-}G(\mathcal T)$ and $G$ is submodular.

To prove monotonicity it is enough to show $G(\mathcal S{\cup}\left\lbrace x \right\rbrace ){-}G(\mathcal S)\geq 0$. This time expand the mutual information according to $I(x;y)=H(y){-}H(y\vert x)$:
\begin{align}
G(\mathcal S{\cup}\left\lbrace x \right\rbrace ){-}G(\mathcal S)=H(\boldsymbol{\beta}  \vert \hat{\boldsymbol{f}}_{\mathcal S}){-}H(\boldsymbol{\beta}  \vert \hat{\boldsymbol{f}}_{\mathcal S\cup\left\lbrace x \right\rbrace})
\end{align} 
Conditioning can never increase entropy so $H(\boldsymbol{\beta}  \vert \hat{\boldsymbol{f}}_{\mathcal S}) {\geq} H(\boldsymbol{\beta}  \vert \hat{\boldsymbol{f}}_{\mathcal S\cup\left\lbrace x \right\rbrace})$ and the result follows.
\end{proof}
\subsubsection{Efficient Solvers}
Our optimization problem \eqref{eq:trunc_disc_opt} is the maximization of a monotonic submodular function. A greedy algorithm, shown as Algorithm \ref{al:greedy_max}, solves this problem to within the best possible approximation factor, as the following fundamental theorem states.
\begin{algorithm}
	\caption{Greedy Submodular Maximization}
	\label{al:greedy_max}
	\begin{algorithmic}[1]
		\Function{GreedyMax} {$G(\cdot),\mathcal{V},N$}
		\State $\mathcal S\leftarrow \emptyset$
		\For{$i=1$ to $N$}
		\State $x^{\star}=\text{argmax}_{x\in \mathcal{V}\setminus \mathcal{S}}\,G(\mathcal S\cup \left\lbrace x \right\rbrace )$
		\State $\mathcal S \leftarrow \mathcal S \cup \left\lbrace x^{\star}\right\rbrace $
		\EndFor
		\State Return $\mathcal S$
		\EndFunction
	\end{algorithmic}
\end{algorithm}\\
The following theorem guarantees that the greedy algorithm achieves approximately optimal performance:
\begin{theorem}\label{thm:nemhauser}{(Nemhauser \cite{nemhauser1978analysis})} Let $G$ be a monotonic, submodular set function and $\mathcal{S}^{\star}=\underset{\mathcal S\subseteq \mathcal{V},\vert \mathcal S\vert\leq N}{\text{argmax}}\,G(\mathcal S)$.\\ Let $\mathcal{S}^{\text{gr}}$ be the set retrieved by the greedy maximization Algorithm \ref{al:greedy_max}. We have the following guarantee for the performance of the greedy algorithm:
$$G(\mathcal{S}^{\text{gr}})\geq(1-(1-\frac{1}{N})^{N})G(\mathcal{S}^{\star})\geq (1-\frac{1}{e})G(\mathcal{S}^{\star})$$
	Moreover, no polynomial time algorithm can provide a better approximation guarantee
	unless P=NP \cite{feige1998threshold}.
\end{theorem} 
Combining the guarantees of Theorem \ref{thm:nemhauser}, Lemma \ref{lemma:1} and \ref{lemma:2} we derive an approximation bound on the original problem \eqref{eq:optimization_problem}:
\begin{corollary}
	\begin{align}
		(1{-}\frac{1}{e})\left[I(\tilde{\boldsymbol{f}}_{\mathcal{S}^{\star}}  ; \left\lbrace \beta_m\right\rbrace){-}N\log\frac{\lambda\sigma_w^2{+}4\delta P(1{+}\delta)N^{\frac32}}{\lambda\sigma_w^2{-}\epsilon\lambda N^{\frac32}}\right] {\leq} I(\hat{\boldsymbol{f}}_{\mathcal{S}^{\text{gr}}}  ; \boldsymbol{\beta})
	\end{align}
\end{corollary}
We hence apply Algorithm \ref{al:greedy_max} to solve our optimization problem \eqref{eq:trunc_disc_opt}. The algorithm runs in time $O(\vert \mathcal{V}\vert N)$, linear in the size of the set $\mathcal{V}$ and the number of selected elements $N$ \cite{minoux1978accelerated} such that it is easily implementable for problems of large size. However, as it turns out even more efficient variants of the algorithm have been introduced and studied. One such variant commonly referred to as the 'lazy greedy' algorithm was studied in \cite{minoux1978accelerated} and it was shown to offer substantial running-time improvements in practice (with an unlikely worst-case theoretical performance upper bounded by that of the conventional greedy algorithm). Our numerical experiments described in Section \ref{sec:numerical_experiments} implement this more efficient variant to reduce running time.
\subsubsection{Improved Approximation Bounds}
While Theorem \ref{thm:nemhauser} guarantees an approximation bound of $(1-\frac{1}{e})\approx63\%$ for the efficient greedy algorithm this guarantee is not tight. It is possible to derive a tighter data-dependent online bound on the gap between the cost of the greedy solution $G(\mathcal{S}^{\text{gr}})$ and that of the optimal solution $G(\mathcal{S}^{\star})$ \cite{leskovec2007cost}. Specifically, we have:
\begin{align} \label{eq:online_bound}
G(\mathcal{S}^{\text{gr}}){\geq} G(\mathcal{S}^{\star}){-}\max_{B:\vert B\vert \leq N}\sum_{e\in B} \left( G(\mathcal{S}^{\text{gr}}{\cup}\left\lbrace e\right\rbrace ){-}G(\mathcal{S}^{\text{gr}})\right) 
\end{align}
which takes $O(\vert \mathcal{V}\vert\log \vert \mathcal{V}\vert)$ evaluations of $G(\mathcal S)$ to compute and sort. We use \eqref{eq:online_bound} to improve the distance from optimality bound in some of our numerical solutions in Section \ref{sec:numerical_experiments}.
\subsection{Design Example: A Simple Ideal Setting} \label{sec:design_example}
In the previous subsections we formulated the array design problem in a setting with constraints on the aperture $\mathcal A$ and the number of sensors $N$ and showed how a greedy algorithm (Algorithm \ref{al:greedy_max}) is guaranteed to efficiently find an approximate solution.

Here, we study a particular instance of that problem, where the Signal to Noise Ratio (SNR) is high, and the aperture is effectively unconstrained (the $N$ sensors may be placed anywhere on the real line). Under these conditions a truncated $\frac{\lambda}{2}$-spaced array is traditionally considered the design of choice in the conventional non-Bayesian setting (this is a truncated version of the infinite $\frac{\lambda}{2}$-spaced design mentioned in Section \ref{sec:far_field_imaging}). We show next that the truncated $\frac{\lambda}{2}$-spaced design naturally emerges as the approximately optimal solution as retrieved by our schemes in Bayesian settings where the a-priori $\boldsymbol{\beta}$ distribution satisfies some conditions. Specifically, we have the following result:
\begin{theorem} \label{thm:uniform_array}
	Consider the high SNR regime $\frac{P}{\sigma_w^2}\rightarrow \infty $ and assume the prior from \eqref{eq:prior_gaussian} takes a symmetric, monotonically decreasing form, i.e. $\sigma_m^2=\sigma_{-m}^2$ and $\sigma_{m_1}^2\geq \sigma_{m_2}^2$ whenever $0\leq m_1<m_2$. In addition, take $\mathcal{V}$ as an arbitrarily dense set of sampling points on $\mathbb{R}$, and $\mathcal{M}=-M,\ldots,M$ with $M\rightarrow \infty$.\\
	We then have that a greedy solver on \eqref{eq:trunc_disc_opt} will return a length $N$, $\frac{\lambda}{2}$-spaced truncated uniform array centered around $x=0$. 
\end{theorem}
\begin{proof}
	 See Appendix \ref{app:uniform_array}.
\end{proof}
The last theorem studies one class of simple idealized problems where the greedy solution is reminiscent of generic non-Bayesian array designs. However, notice that our formalism is also useful in more challenging design problems such as when the aperture $\mathcal A$ takes on arbitrary forms, and the effects of noise and application-tailored priors are considered. Furthermore, in Section \ref{sec:matroid_constraints} we demonstrate that additional combinatorial constraints may be naturally incorporated into our formulation such that even more challenging problems become manageable.
\section{Array Design with Matroid Constraints} \label{sec:matroid_constraints}
In Section \ref{sec:array_design} we formalized the array design problem in a setting where we imposed constraints on the aperture and the number of sensors. Specifically, the constraints were $\mathcal S\subseteq \mathcal{V},\,\vert \mathcal S\vert \leq N$. In many practical scenarios these may be too simplistic to accurately represent real world design constraints. For example in applications where sensors are heavy and mounted on support beams we may want to restrict the number of sensors in specific sections of the aperture.

In this section we briefly review key elements from matroid theory which is a branch in combinatorics \cite{oxley2006matroid} and survey results from submodular optimization with matroid constraints guaranteeing the existence of efficient approximate solvers for this class of problems. We continue to show that these mathematical structures may be utilized to impose constraints of interest in array design enriching the set of problems our Bayesian formulation can describe and solve.
\subsection{Submodular Optimization with Matroid Constraints}
We begin by defining matroids and their corresponding independent sets \cite{nemhauser1978analysis}:
\begin{definition} \label{sec:matroid_def_opt}
	A finite \textbf{matroid} $M$ is a pair $(\mathcal{V},\mathcal{I})$ where $\mathcal{V}$ is a ground set and $\mathcal{I}$ is a collection of subsets of $\mathcal{V}$ (the \textbf{independent sets}) that satisfies the following properties:
	\begin{enumerate}
		\item The empty set is independent: $\emptyset\in \mathcal{I}$
		\item A subset of an independent set is independent:\\ $X\subset Y,\, Y\in \mathcal{I}\quad \Rightarrow \quad X\in \mathcal {I}$
		\item If $X$ is an independent set and $Y$ is a larger indepedent set, $X$ can be augmented to a larger independent set by adding an element from $Y{\setminus} X$:\\
		$X,Y{\in} \mathcal{I},\, \vert X\vert{<}\vert Y\vert \quad \Rightarrow \quad \exists e\in Y{\setminus} X \,\,\,\, \text{s.t.} \,\,\,\, X{\cup}\left\lbrace e \right\rbrace \in \mathcal{I}$
	\end{enumerate}
\end{definition}
A matroid structure may be used to classify subsets of a ground set $\mathcal{V}$ into permissible subsets which belong to $\mathcal{I}$ and non permissible subsets which do not belong to $\mathcal{I}$. In the next subsection we show that using this formalism we can express interesting array design constraints.

From the theory of submodular optimization we have the following results for submodular optimization with matroid constraints \cite{krause2012submodular}, \cite{calinescu2011maximizing}. Let $M=(\mathcal{V},\mathcal{I})$ be a matroid and $G(\mathcal S)$ a monotonic, submodular set function. There exists an efficient approximate solver for the problem $\text{argmax}_{\mathcal S\in \mathcal{I}}\, G(\mathcal S)$. Specifically, a greedy solver (maximizing the immediate marginal benefit at each step) taking the form 
\begin{align}
\mathcal{S}^{\text{gr}}{\leftarrow} \mathcal{S}^{\text{gr}} \cup \left\lbrace \underset{e:e\notin \mathcal{S}^{\text{gr}},\mathcal{S}^{\text{gr}}\cup \left\lbrace e\right\rbrace \in \mathcal{I}}{\text{argmax}}\left[ G(\mathcal{S}^{\text{gr}}{\cup} \left\lbrace e\right\rbrace){-}G(\mathcal{S}^{\text{gr}})\right] \right\rbrace 
\end{align}
and stopping when no more elements $e$ can be added is guaranteed to achieve a half-approximation bound: 
\begin{align}
G(\mathcal{S}^{\text{gr}})\geq \frac12\, \underset{\mathcal S\in \mathcal{I}}{\text{max}} \,G(\mathcal S).
\end{align}
The constant factor may be tightened to $\left( 1-\frac{1}{e}\right)$ by utilizing specialized randomized algorithms \cite{calinescu2011maximizing}.
\subsection{Combinatorial Constraints in Array Design} \label{sec:comb_con_arr_des}
Here we invoke one specific well known matroid structure and show its application in expressing useful array design constraints. Let $\mathcal{V}$ be a ground set of grid points where sensors are allowed to be placed as before. Let $\mathcal{V}_1,\ldots,\mathcal{V}_K$ be a partition of the set $\mathcal{V}$, i.e. $\,\,\, \bigcup\limits_k\mathcal{V}_k=\mathcal{V}, \,\,\,\, \mathcal{V}_i \bigcap \mathcal{V}_j = \emptyset, \, \forall i\neq j$, and let $N,n_1,\ldots,n_K$ be a set of integers.

We define the (cardinality constrained) partition matroid \cite{calinescu2011maximizing}  $M=(\mathcal{V},\mathcal{I})$ with the following definition for the collection of independent sets: a subset $\mathcal S {\subseteq} \mathcal V$ is an independent subset $\mathcal S{\in}\mathcal{I}$ if it holds $\vert \mathcal S{\cap} \mathcal{V} \vert \leq N,\,\vert \mathcal S{\cap} \mathcal{V}_j \vert \leq n_j ,\,\, \forall j$.

In the context of array design the partition matroid may be useful in expressing practical constraints over sensor placement configurations. For example if the subsets $\mathcal{V}_i$ represent closed line sections, e.g. a physical partitioning of the aperture into zones, and $\mathcal{I}$ represents the collection of all permissible designs then the structure of the matroid limits the number of sensors that may be placed in the $i$\textit{th} zone to $n_i$ which may be an important engineering constraint coupled with some specific application. We solve:
\begin{align}
\mathcal{S}^{\star}=\text{argmax}_{\mathcal{S}\in \mathcal{I}}\, I(\hat{\boldsymbol{f}}_{\mathcal S}  ; \boldsymbol{\beta})
\end{align}
Applying the results from the previous subsection we immediately have an efficient approximate solver for the array design problem coupled with a partition matroid constraint. In Section \ref{sec:numerical_experiments} we detail such a design for one numerical example.

\section{Numerical Experiments} \label{sec:numerical_experiments}
In this section we perform numerical experiments validating our theoretical results and exemplifying them. We showcase an array design with cardinality and aperture constraints as prescribed in Section \ref{sec:array_design} and design arrays adhering to matroid constraints as prescribed in Section \ref{sec:matroid_constraints}.

Our initial setting is as follow. We fix $\lambda=1$ throughout as the wavelength only serves to scale the $x$ axis. The aperture is set as $\mathcal A=\left\lbrace x\vert-3.5\leq x \leq 3.5\right\rbrace $ and the selection set $\mathcal{V}$ is chosen as a uniform grid of $113$ positions from $\mathcal A$ spaced $\delta=0.0625$ apart.
We set out to design an array consisting of $N=11$ sensor locations. The prior for $\left\lbrace \beta_m \right\rbrace $ is set as per \eqref{eq:prior_def} with $r=1$ and normalized to sum to $P=1$. For the simulations we consider the truncated vector $\boldsymbol{\beta}$ formed when restricting the set of $m$ coefficients to $901$ consecutive elements centered around the origin, i.e., we set $\mathcal{M}=\left\lbrace{-450},\ldots,{+}450\right\rbrace $. For the preliminary design we implement the lazy greedy algorithm of Section \ref{sec:array_design} and plot the results in the left column of Fig.~\ref{fig:cont_aperture} as a function of the SNR which we define here as $\frac{P}{N\sigma_w^2}$. Blue markers denote the full selection set $\mathcal{V}$ and red markers delineate the active $\mathcal S$ selected by the algorithm.\\ 
\begin{figure*}[!t]
	\centering
	\includegraphics[width=1\textwidth]{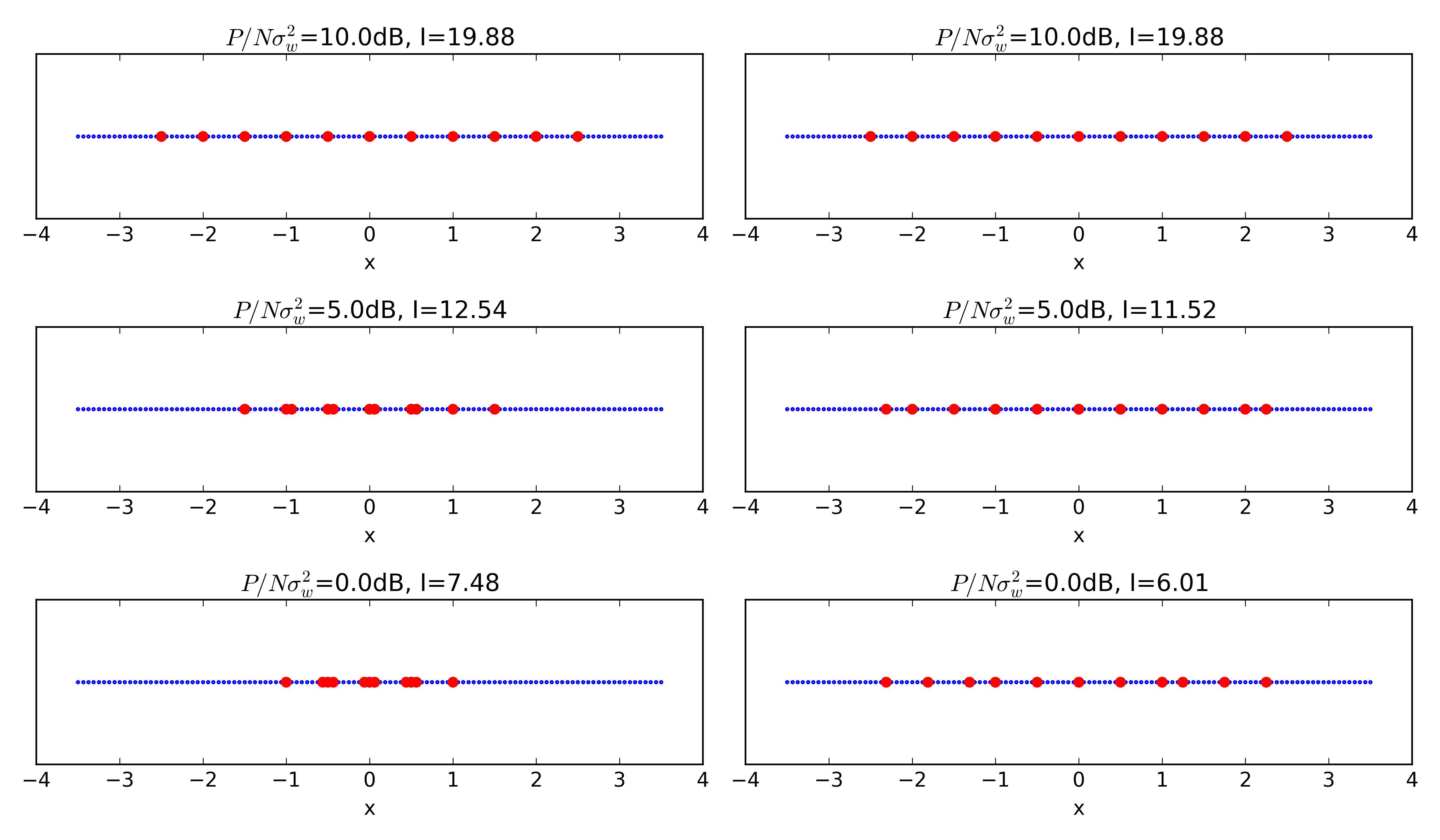}
	\caption{(left) Array designs with an aperture constraint and various SNR levels. (right) Array designs with combinatorial placement constraints and various SNR levels.} \label{fig:cont_aperture}	
\end{figure*}
In the high SNR regime ($\text{SNR}=10\text{dB}$ or higher values) the resulting design is a truncated $\frac{\lambda}{2}$ uniform array as predicted according to Theorem \ref{thm:uniform_array}. As the SNR decreases the reliability of the measurements deteriorates and the algorithm prefers locating samplers right next to each other on expense of widening the array as this serves to average out the noise.

The performance in terms of mutual information $I(\hat{\boldsymbol{f}}_{\mathcal S}  ; \boldsymbol{\beta})$ for the selected locations $\mathcal S$ appears in the title of the plots (in natural units). Notice for example that for the $5\text{dB}$ SNR design the achieved mutual information is $12.54$. Using Theorem \ref{thm:nemhauser} we have that the optimal design cannot achieve mutual information better than $\frac{1}{1{-}\frac{1}{e}}12.54=19.83$. This bound can be improved using the improved bounding method briefly described following Theorem \ref{thm:nemhauser} to show that the optimal performance is not greater than $17.45$.

The truncation level dictated by our choice of $\mathcal{M}$ translates to $\epsilon\cong 1e{-}4$ and the truncation bounds from Lemma \ref{lemma:1} read (for the lower, extreme SNR case): ${-}0.45 {\leq} I(\tilde{\boldsymbol{f}_{\mathcal S}}  ; \left\lbrace \beta_m\right\rbrace )-I(\hat{\boldsymbol{f}_{\mathcal S}}  ; \boldsymbol{\beta} ) {\leq} 0.47$. We find empirically that these bounds are extremely loose and $\mathcal{M}$ can be shrunk considerably without substantially compromising accuracy. To achieve a similar upper bound on $I(\hat{\boldsymbol{f}}_{\mathcal{S}^{\star}}  ; \boldsymbol{\beta})-I(\hat{\boldsymbol{f}}_{\mathcal{S}^{\star}_d}  ; \boldsymbol{\beta})$ as per Lemma \ref{lemma:2} a discretization level of $\delta\cong 2e{-}4$ is needed. However, we empirically find that our choice of $\delta=0.0625$ is accurate enough as further refining the grid does not significantly change the design. Our lemmas prove to be extremely pessimistic as is expected given that the proofs take into account worst-case scenarios.

The array geometries above, derived according to the formulations of Section \ref{sec:array_design}, are designed to optimize the quality of inference between the measurements and the scene expansion coefficients $\boldsymbol{\beta}$. Many sensing applications of interest specifically involve imaging the scene, that is reconstructing $\beta(\psi)$ from the measurements. Our next experiment was designed to empirically evaluate the Mean Square Error (MSE) performance in scene reconstruction from measurements collected using the prescribed designs. 
First, we designed five array geometries as described above, optimized for several target SNR levels $\left\lbrace30\text{dB},12\text{dB},10\text{dB},5\text{dB},0\text{dB} \right\rbrace$.
We set up a Monte-Carlo experiment where $1000$ scenes were randomly drawn from the distribution of Section \ref{secsub:prior}. For each scene, noisy measurements were collected by each of the five optimized arrays. The measurements were repeated with five different synthetic noise levels corresponding to the five target SNR levels.

We repeatedly performed maximum likelihood estimation \cite{adali2011complex} of the expansion coefficients $\boldsymbol \beta$, and synthesized an estimated scene $\hat{\beta}(\psi)$ according to \eqref{eq:beta_f_expansion}. The MSE discrepancy between $\hat{\beta}(\psi)$ and the true scene is depicted in Fig.~\ref{fig:mse_comparison}. It is evident that the quality of inference criterion is indicative of MSE performance, as each of the five geometries yielded the best MSE performance at its specified target SNR level.\\
\begin{figure}[t]
	\centering
	\includegraphics[width=0.5\textwidth]{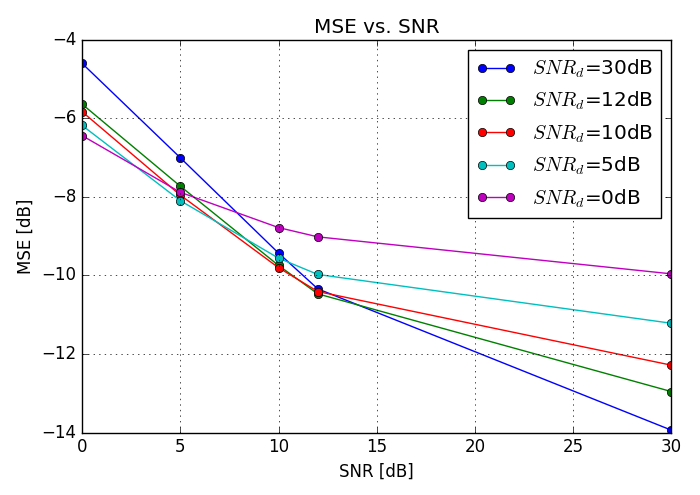}
	\caption{Reconstruction MSE for various arrays designed for fixed SNR levels, with actual SNR levels swept.} \label{fig:mse_comparison}	
\end{figure}
Next we solve a corresponding set of design problems with matroid constraints installed to limit the number of sensors in given aperture segments. Specifically, we use the (cardinality constrained) partition matroid from Section \ref{sec:comb_con_arr_des} with $N=11, n_i=1, \forall i$ and $\mathcal{V}_i$ spanning consecutive line segments of length $0.5$: $\mathcal{V}_i{=}[{-}0.25{+}0.5{\cdot} i,{+}0.25{+}0.5{\cdot} i)$. The matroid constraints limits the proximity between sensor elements, which may be useful in real world use cases. We plot the results for the matroid constrained designs in the right column of Fig.~\ref{fig:cont_aperture}. Notice that while the theoretical guarantees pertaining to the greedy matroid optimization scheme of Section \ref{sec:matroid_constraints} is $0.5$ compared to $1-\frac{1}{e}$ with the cardinality constraints of Section \ref{sec:array_design}, the actual performance achieved in the constrained design instances is not far from those achieved with the simple cardinality constraints.
\section{Concluding remarks} \label{sec:conclusions}
We introduced a novel framework for designing sensor arrays. Our setting is Bayesian in the sense that our model incorporates the notion of prior beliefs about the scene of interest. Our goal is designing arrays that are adapted to the scene and perform efficient inference. We showed that this NP-hard combinatorial selection problem may be efficiently approximated by porting results from the theory of submodular set function optimization.

Initially, we showed how to apply our formalism in design problems with straightforward cardinality and aperture constraints. Later we showed that more challenging combinatorial constraints may be enforced by utilizing results and concepts from the field of matroids and submodular optimization with matroid constraints.

Our formulation connecting the array design problem and submodularity may further be extended to reap additional benefits, namely tackling problems such as robust array design or adaptive design that evolve as the scene is learned. We leave the treatment of some of these and other problems to separate study \cite{shulkind2017mulwav}.
\appendices
\section{Proof of Lemma \ref{lemma:1}} \label{app:proof_of_truncation_lemma}
For simplicity of notation, we suppress the subscript $\mathcal S$ throughout. Begin by expanding the mutual information expressions:  
	\begin{align}
	I(\tilde{\boldsymbol{f}}  ; \left\lbrace \beta_m\right\rbrace )&=H(\tilde{\boldsymbol{f}})-H(\tilde{\boldsymbol{f}}\vert \left\lbrace \beta_m\right\rbrace)	\notag\\
	I(\hat{\boldsymbol{f}}  ; \boldsymbol{\beta} )&=H(\hat{\boldsymbol{f}})-H(\hat{\boldsymbol{f}}\vert \boldsymbol{\beta})	
	\end{align}
	Examining \eqref{eq:full_model} and \eqref{eq:appx_model} we have $H(\tilde{\boldsymbol{f}}\vert \left\lbrace \beta_m\right\rbrace)=H(\hat{\boldsymbol{f}}\vert \boldsymbol{\beta})=H(\boldsymbol{w})$ and so: 
	\begin{align}
	I(\tilde{\boldsymbol{f}}  ; \left\lbrace \beta_m\right\rbrace )-I(\hat{\boldsymbol{f}}  ; \boldsymbol{\beta} )=H(\tilde{\boldsymbol{f}})-H(\hat{\boldsymbol{f}})
	\end{align}
	Next, notice that $\tilde{\boldsymbol{f}}$ and $\hat{\boldsymbol{f}}$ are both circular, complex, Gaussian random $N$-length vectors, such that their entropies are given according to \cite{neeser1993proper}:
	\begin{align} \label{eq:entropies}
	H(\tilde{\boldsymbol{f}})&=\log((\pi e)^{N} \text{det}(\tilde{\boldsymbol \Sigma}))& & \tilde{\Sigma}_{ij}=\mathbb{E}[\tilde{\boldsymbol{f}}_i\tilde{\boldsymbol{f}}_j^{*}]\notag\\
	H(\hat{\boldsymbol{f}})  &=\log((\pi e)^{N} \text{det}(\hat{\boldsymbol \Sigma}))& & \hat{\Sigma}_{ij}=\mathbb{E}[\hat{\boldsymbol{f}}_i\hat{\boldsymbol{f}}_j^{*}]
	\end{align}
	Further expand using the independence between $\boldsymbol w$ and $\beta_m$ and $\mathbb{E}[\beta_m\beta^{*}_{m'}]=\delta_{mm'}\sigma^2_m$:
	\begin{align}
	\tilde{\Sigma}_{ij}&=\mathbb{E}[(\sum\limits_{m}K_{im}\beta_m+w_i)(\sum\limits_{m'}K^{*}_{jm'}\beta^{*}_{m'}+w_j^{*})]\notag\\
	&=[\boldsymbol\Sigma_{ww}]_{ij}+\sum\limits_{m}K_{im}K^{*}_{jm}\sigma^2_m\label{eq:sigma_1}\\ 
	\hat{\Sigma}_{ij}&=\mathbb{E}[(\sum\limits_{m\in \mathcal{M} }K_{im}\beta_m+w_i)(\sum\limits_{m' \in \mathcal{M}}K^{*}_{jm'}\beta^{*}_{m'}+w_j^{*})]\notag\\
	&=[\boldsymbol\Sigma_{ww}]_{ij}+\sum\limits_{m\in \mathcal{M}}K_{im}K^{*}_{jm}\sigma^2_m	\label{eq:sigma_2}
	\end{align}
	Comparing \eqref{eq:sigma_1} and \eqref{eq:sigma_2} and using:  
	\begin{align}
	\left\lvert \sum\limits_{m\notin \mathcal{M}}K_{im}K^{*}_{jm}\sigma^2_m \right\rvert &\leq \sum\limits_{m\notin \mathcal{M}}\left\lvert K_{im}K^{*}_{jm}\sigma^2_m\right\rvert \leq  \sum\limits_{m\notin \mathcal{M}} \sigma^2_m =\epsilon
	\end{align}
	we have $\vert\tilde{\Sigma}_{ij}-\hat{\Sigma}_{ij}\vert \leq \epsilon$ such that we may write:
	\begin{align} \label{eq:cov_pert}
	\hat{\boldsymbol \Sigma}=\tilde{\boldsymbol \Sigma}+\epsilon\boldsymbol X
	\end{align}
	for some $N \times N$ matrix $\boldsymbol X$ satisfying $\left\lvert X_{ij}\right\rvert \leq 1$. We use \eqref{eq:cov_pert} to bound the determinants. $\tilde{\boldsymbol \Sigma}$ is positive-definite and invertible such that we can write: 
	\begin{align} \label{eq:cov_matrix_trunc}
	\text{det}(\hat{\boldsymbol \Sigma})=\text{det}(\tilde{\boldsymbol \Sigma}+\epsilon\boldsymbol X)=\text{det}(\tilde{\boldsymbol \Sigma})\text{det}(\boldsymbol {I}_N+\epsilon \tilde{\boldsymbol \Sigma}^{-1} \boldsymbol X )
	\end{align} 
	Substituting \eqref{eq:cov_matrix_trunc} in \eqref{eq:entropies} we have:
	\begin{align}
	H(\hat{\boldsymbol{f}})&{-}H(\tilde{\boldsymbol{f}})=\log(\text{det}(\boldsymbol{I} _N{+}\epsilon \tilde{\boldsymbol \Sigma}^{-1} \boldsymbol X  )){=}\log(\text{det}(\tilde{\boldsymbol X} ))
	\end{align}
	 with $\tilde{\boldsymbol X}\equiv \boldsymbol I_N+\epsilon \tilde{\boldsymbol \Sigma}^{-1}\boldsymbol X$. We turn next to bounding the term $\log(\text{det}(\tilde{\boldsymbol X} ))$. First notice:
	\begin{align} 
	&\left\lvert[\epsilon \tilde{\boldsymbol \Sigma}^{-1} \boldsymbol X]_{ij}\right\rvert=\epsilon\left\lvert \sum\limits_{m}\tilde{\Sigma}^{-1}_{im}X_{mj}\right\rvert \leq \epsilon \sum\limits_{m} \left\lvert\tilde{\Sigma}^{-1}_{im}X_{mj}\right\rvert \notag\\
	&\leq \epsilon \sum\limits_{m} \left\lvert\tilde{\Sigma}^{-1}_{im}\right\rvert 
	\leq \epsilon \|\tilde{\boldsymbol \Sigma}^{-1} \|_{\infty}\leq \epsilon\sqrt{N} \|\tilde{\boldsymbol \Sigma}^{-1} \|_{2}\notag\\
	&=\epsilon\sqrt{N} \sigma_{\text{max}}(\tilde{\boldsymbol \Sigma}^{-1})=\epsilon\sqrt{N} \frac{1}{\sigma_{\text{min}}(\tilde{\boldsymbol \Sigma})}\leq \frac{\epsilon\sqrt{N}}{\sigma_{w}^{2}}
	\end{align}
	Where we have used the matrix norm equivalence $\| \boldsymbol A\|_{\infty}\leq \sqrt{N} \|\boldsymbol A \|_{2}$ (for $N {\times} N$ matrices) and $\sigma_{\text{max}}(\cdot)$ ($\sigma_{\text{min}}(\cdot)$) is the maximal (minimal) singular value such that $\sigma_{\text{min}}(\tilde{\boldsymbol \Sigma})\geq \sigma_w^2$.
	Thus we have that $\tilde{\boldsymbol X}$ has diagonal elements centered around 1: $\left\lvert\tilde{X}_{ii}-1\right\rvert\leq \frac{\epsilon\sqrt{N}}{\sigma_{w}^{2}} $ and the row-sums over non-diagonal entries satisfy $\sum\limits_{m\neq i} \left\lvert\tilde{X}_{im} \right\rvert \leq \frac{\epsilon\sqrt{N}(N-1)}{\sigma_{w}^{2}}$. 
	
	Applying the Gershgorin circle theorem we have for the eigenvalues of $\tilde{\boldsymbol X} $:
	$$1-\frac{\epsilon\sqrt{N}N}{\sigma_{w}^{2}} \leq \vert\lambda_i\vert  \leq 1+\frac{\epsilon\sqrt{N}N}{\sigma_{w}^{2}}$$
	$\text{det}(\tilde{\boldsymbol X} )$ is a positive real number as the quotient of the determinants of two positive definite matrices	$\text{det}(\tilde{\boldsymbol X} )=\frac{\text{det}(\hat{\boldsymbol \Sigma})}{\text{det}(\tilde{\boldsymbol \Sigma})}$
	such that we can write $\text{det}(\tilde{\boldsymbol X} )=\prod\nolimits_i \lambda_i=\prod\nolimits_i \vert \lambda_i \vert$ and consequently:
	\begin{align}
	N\log(1-\frac{\epsilon N^{\frac32}}{\sigma_{w}^{2}})  \leq\log(\text{det}(\tilde{\boldsymbol X} ))\leq N\log(1+\frac{\epsilon N^{\frac32}}{\sigma_{w}^{2}})
	\end{align}
	This finally leads to:
	\begin{align}
	 {-} N\log(1{+}\frac{\epsilon  N^{\frac32}}{\sigma_{w}^{2}}) \leq I(\tilde{\boldsymbol{f}}  ; \left\lbrace \beta_m\right\rbrace ) {-}I(\hat{\boldsymbol{f}}  ; \boldsymbol{\beta} )\leq{-}N\log(1{-}\frac{\epsilon  N^{\frac32}}{\sigma_{w}^{2}}) 
	\end{align}
\section{Proof of Lemma \ref{lemma:2}} \label{app:proof_of_discretization_lemma}
The left inequality is trivial. We have $I(\hat{\boldsymbol{f}}_{\mathcal{S}_d^{\star}}  ; \boldsymbol{\beta}) \leq I(\hat{\boldsymbol{f}}_{\mathcal{S}^{\star}}  ; \boldsymbol{\beta})$ as the second optimization is over a larger set.

To prove the right inequality we show that for every $\mathcal{S}\subseteq \mathcal A$ there is $\mathcal{S}_d \subseteq \mathcal{V}$ such that 
  \begin{align} \label{eq:upper_bound}
  I(\hat{\boldsymbol{f}}_{\mathcal S}  ; \boldsymbol{\beta})\leq I(\hat{\boldsymbol{f}}_{\mathcal{S}_d}  ; \boldsymbol{\beta}) {+} N\log(1{+}\frac{4\delta P (1{+}\delta) N^{\frac32}}{\lambda\sigma_{w}^{2}})
  \end{align}
we will show this for $\vert \mathcal S \vert {=} N$ but the proof for other cardinalities is identical.

With distance $\delta$ between adjacent elements of $\mathcal{V}$, for every $\mathcal S{=}\left\lbrace x_1,\ldots,x_{N} \right\rbrace \subset \mathcal A$ there is a set $\mathcal{S}_d{=}\left\lbrace x_1^d,\ldots,x_{N}^d\right\rbrace$ such that $\mathcal{S}_d \subseteq \mathcal{V}$ and $\vert x_i - x_i^d\vert \leq \frac{\delta}{2}$ for all $i$. We have, similarly to Appendix \ref{app:proof_of_truncation_lemma}: 
\begin{align}
I(\hat{\boldsymbol{f}}_{\mathcal S}  ; \boldsymbol{\beta})-I(\hat{\boldsymbol{f}}_{\mathcal{S}_d};\boldsymbol{\beta})=H(\hat{\boldsymbol{f}}_{\mathcal S})-H(\hat{\boldsymbol{f}}_{\mathcal {S}_d})
\end{align}
Using the model \eqref{eq:appx_model}:
\begin{align} 
H(\hat{\boldsymbol{f}}_{\mathcal S})&=\log((\pi e)^{N} \text{det}(\hat{\boldsymbol \Sigma}^{\mathcal S}))\notag\\
H(\hat{\boldsymbol{f}}_{\mathcal{S}_d}) &=\log((\pi e)^{N} \text{det}(\hat{\boldsymbol \Sigma}^{\mathcal{S}_d})) 
\end{align}
where:
\begin{align} \label{eq:cov_mat_def}
\hat{\boldsymbol \Sigma}^{\mathcal S}&\equiv\mathbb{E}[\hat{\boldsymbol{f}}_{\mathcal S}\hat{\boldsymbol{f}}_{\mathcal S}^{\dagger}] =\boldsymbol{K}_{\mathcal S}\boldsymbol{\Sigma}_{\beta\beta}\boldsymbol{K}_{\mathcal S}^{\dagger}+\boldsymbol{\Sigma}_{ww}\notag\\
\hat{\boldsymbol\Sigma}^{\mathcal{S}_d}&\equiv\mathbb{E}[\hat{\boldsymbol{f}}_{\mathcal{S}_d}\hat{\boldsymbol{f}}_{\mathcal{S}_d}^{\dagger}] =\boldsymbol{K}_{\mathcal{S}_d}\boldsymbol{\Sigma}_{\beta\beta}\boldsymbol{K}_{\mathcal{S}_d}^{\dagger}+\boldsymbol{\Sigma}_{ww}\notag\\
\boldsymbol{\Sigma}_{\beta\beta}&\equiv\mathbb{E}[\boldsymbol{\beta}\boldsymbol{\beta}^{\dagger}]
\end{align}
	and $H(\hat{\boldsymbol{f}}_{\mathcal S})-H(\hat{\boldsymbol{f}}_{\mathcal {S}_d})=\log(\text{det}(\hat{\boldsymbol \Sigma}^{\mathcal{S}}))-\log(\text{det}(\hat{\boldsymbol \Sigma}^{\mathcal{S}_d}))$.\\
Both $\boldsymbol{K}_{\mathcal{S}}$ and $\boldsymbol{K}_{\mathcal{S}_d}$ are size $N\times\vert \mathcal{M}\vert$ matrices defined as per the definition in \eqref{eq:K_definition}, such that:
\begin{align}
&\vert[\boldsymbol{K}_{\mathcal{S}}]_{nm}{-}[\boldsymbol{K}_{\mathcal{S}_d}]_{nm}\vert=\vert \text{sinc}(m{+}\frac{2}{\lambda}x_n){-}\text{sinc}(m{+}\frac{2}{\lambda}x_n^d)\vert\notag\\
&\leq 2\frac{2}{\lambda}\vert x_q{-}x_q^d \vert\leq \frac{2}{\lambda}\delta
\end{align}
where for the first inequality we have used the fact that $\text{sinc}(\cdot)$ is Lipschitz with constant smaller than $2$.
We can thus define $\boldsymbol\Delta\equiv\boldsymbol{K}_{\mathcal{S}}-\boldsymbol{K}_{\mathcal{S}_d}$ and we have $\vert\Delta_{nm}\vert\leq \frac{2}{\lambda}\delta$.
Substitution in \eqref{eq:cov_mat_def} yields:
\begin{align}
\hat{\boldsymbol \Sigma}^{\mathcal{S}}=\hat{\boldsymbol\Sigma}^{\mathcal{S}_d}+\boldsymbol\Delta  \boldsymbol{\Sigma}_{\beta\beta} \boldsymbol\Delta^{\dagger} + \boldsymbol\Delta \boldsymbol{\Sigma}_{\beta\beta} \boldsymbol{K}_{\mathcal{S}_d}^{\dagger} + \boldsymbol{K}_{\mathcal{S}_d}\boldsymbol{\Sigma}_{\beta\beta}\boldsymbol\Delta^{\dagger}
\end{align}
We bound the perturbation terms by noticing:
\begin{align}
&\left\lvert [\boldsymbol\Delta  \boldsymbol{\Sigma}_{\beta\beta} \boldsymbol\Delta^{\dagger} ]_{ij} \right\rvert =\left\vert\sum\limits_m\Delta_{im}[\boldsymbol{\Sigma}_{\beta\beta}]_{mm}\Delta_{jm}\right\vert \notag\\ &\leq(\frac{2}{\lambda}\delta)^2\sum\limits_m[\boldsymbol{\Sigma}_{\beta\beta}]_{mm}
\leq \frac{4}{\lambda^2}\delta^2 P\notag\\
&\left\lvert [\boldsymbol\Delta  \boldsymbol{\Sigma}_{\beta\beta} \boldsymbol{K}_{\mathcal{S}_d}^{\dagger} ]_{ij} \right\rvert =\left\vert\sum\limits_m\Delta_{im}[\boldsymbol{\Sigma}_{\beta\beta}]_{mm} [\boldsymbol{K}_{\mathcal{S}_d}]_{jm}\right\vert \notag\\ &\leq\frac{2}{\lambda}\delta\cdot 1\cdot\sum\limits_m[\boldsymbol{\Sigma}_{\beta\beta}]_{mm}\leq\frac{2}{\lambda}\delta P
\end{align}
 and overall we have:
 \begin{align}
 \vert\hat{\Sigma}^{\mathcal S}_{ij}-\hat{\Sigma}^{\mathcal{S}_d}_{ij}\vert \leq \frac{4}{\lambda}\delta P+\frac{4}{\lambda^2}\delta^2 P=\frac{4}{\lambda}\delta P (1+\delta)
 \end{align}
 define: $\epsilon'\equiv \frac{4}{\lambda}\delta P (1+\delta)$ and we have 
\begin{align}
\hat{\boldsymbol \Sigma}^{\mathcal S}=\hat{\boldsymbol \Sigma}^{\mathcal{S}_d}+\epsilon'\boldsymbol X
\end{align}
 with $N \times N$ matrix $\boldsymbol X$ satisfying $\left\lvert X_{ij}\right\rvert \leq 1$ which is akin to \eqref{eq:cov_pert}. We thus port the results from Appendix \ref{app:proof_of_truncation_lemma} here (we only need the lower bound):
 \begin{align}
 - N\log(1+\frac{\epsilon'  N^{\frac32}}{\sigma_{w}^{2}})\leq I(\hat{\boldsymbol{f}}_{\mathcal{S}_d}  ; \boldsymbol{\beta})-I(\hat{\boldsymbol{f}}_{\mathcal S}  ; \boldsymbol{\beta})
 \end{align}
which, upon substitution of $\epsilon'$ is equivalent to \eqref{eq:upper_bound}.
  
\section{Proof of Theorem \ref{thm:uniform_array}} \label{app:uniform_array}
The greedy algorithm sequentially selects elements according to the rule $x^{\star}=\text{argmax}_{x\in \mathcal{V}\setminus \mathcal{\mathcal S}}\,I(\hat{\boldsymbol{f}}_{\mathcal S{\cup} \left\lbrace x \right\rbrace}  ; \boldsymbol{\beta})$ where $\mathcal S$ is the set of elements selected so far. We recursively show that the added elements can be selected on a $\frac{\lambda}{2}$-spaced grid centered around $x=0$. Expanding the mutual information as in Appendix \ref{app:proof_of_truncation_lemma} we have:
\begin{align}
\text{argmax}_{x\in \mathcal{V}\setminus \mathcal{S}}\,I(\hat{\boldsymbol{f}}_{\mathcal S{\cup} \left\lbrace x \right\rbrace}  ; \boldsymbol{\beta}){=}\text{argmax}_{x\in \mathcal{V}\setminus \mathcal{S}}\, H(\hat{\boldsymbol{f}}_{\mathcal S{\cup} \left\lbrace x \right\rbrace})
\end{align}
We begin by showing that the first selected element is $x_1=0$. Indeed, using the results from Appendix \ref{app:proof_of_truncation_lemma}:
\begin{align}
\underset{x\in \mathcal{V}}{\text{argmax}}\, H(\hat{\boldsymbol{f}}_{\left\lbrace x \right\rbrace})=\underset{x\in \mathcal{V}}{\text{argmax}}\,\text{det}(\hat{\boldsymbol \Sigma}_{11})=\underset{x\in \mathcal{V}}{\text{argmax}}\,\hat{\boldsymbol \Sigma}_{11}
\end{align}
where again using Appendix \ref{app:proof_of_truncation_lemma} and under the assumptions of the theorem (high SNR):
\begin{align} \label{eq:sigma_11_opt}
\hat{\boldsymbol \Sigma}_{11}&=\sum\limits_{m}K_{1m}K^{*}_{1m}\sigma^2_m=\sum\limits_{m}\text{sinc}^2(m{+}\frac{2}{\lambda}x)\sigma^2_m\notag\\
&\leq\sigma^2_0\sum\limits_{m}\text{sinc}^2(m{+}\frac{2}{\lambda}x)=\sigma^2_0
\end{align}
where we used $\sigma_m^2\leq \sigma_0^2,\quad \forall m\neq 0$ and the identity:
\begin{align}
\sum\limits_m \text{sinc}(m{+}a) \text{sinc}(m{+}b)=\text{sinc}(b{-}a)
\end{align}
It easy to see that in \eqref{eq:sigma_11_opt} equality is achieved for the choice $x_1=0$ which is the claim.

Next, assume that the greedy algorithm has already picked a set $\mathcal S=\left\lbrace x_1,\ldots,x_{\vert \mathcal S \vert}\right\rbrace $ of adjacent elements on a $\frac{\lambda}{2}$-spaced grid centered around $x=0$ and show that the next element to be added is an adjacent location on the same $\frac{\lambda}{2}$-spaced grid. We have using $H(x,y){=}H(x){+}H(y\vert x)$:
\begin{align}
\underset{x\in \mathcal{V}\setminus \mathcal{S}}{\text{argmax}}\, H(\hat{\boldsymbol{f}}_{\mathcal S{\cup} \left\lbrace x \right\rbrace}){=}\underset{x\in \mathcal{V}\setminus \mathcal{S}}{\text{argmax}}\, H(\hat{\boldsymbol{f}}_{\left\lbrace x \right\rbrace}\vert \hat{\boldsymbol{f}}_{\mathcal S}){=}\underset{x\in \mathcal{V}\setminus \mathcal{S}}{\text{argmax}}\, \sigma_{x\vert \mathcal S}^2
\end{align}
where $\sigma_{x\vert \mathcal S}^2$ is the conditional variance of the Gaussian observation collected at $x$ given the Gaussian observations made at the set $\mathcal S$:
\begin{align} \label{eq:cov_calc_cond}
\sigma_{x\vert \mathcal S}^2=\sigma_{x}^2-\boldsymbol\Sigma_{x \mathcal S}\boldsymbol\Sigma_{\mathcal S \mathcal S}^{-1}\boldsymbol\Sigma_{x \mathcal S}^{\dagger}
\end{align}
with the usual definitions for the covariance matrices (as in Appendix \ref{app:proof_of_truncation_lemma}):
\begin{align}
[\boldsymbol\Sigma_{x \mathcal S}]_{1i}&{=}\sum\limits_m \text{sinc}(m{+}\frac{2}{\lambda}x)\text{sinc}(m{+}\frac{2}{\lambda}x_i)\sigma^2_m{=}\text{sinc}(m(i){+}\frac{2}{\lambda}x)\sigma_{m(i)}^2\notag\\
[\boldsymbol\Sigma_{\mathcal S \mathcal S}]_{ij}&{=}\sum\limits_{m}\text{sinc}(m{+}\frac{2}{\lambda}x_i)\text{sinc}(m{+}\frac{2}{\lambda}x_j)\sigma^2_m{=}\delta_{ii}\sigma_{m(i)}^2
\end{align}
where in the last equations almost all $\text{sinc}(\cdot)$ functions nulled out as the $x_i$'s are situated on a $\frac{\lambda}{2}$ grid, and we have defined $m(i){\equiv} {-}\frac{2}{\lambda}x_i$, such that $I\equiv\left\lbrace m(i)\right\rbrace$ is a set of consecutive integers.
Additionally, we have:
\begin{align}
\sigma_x^2=\sum\limits_m \text{sinc}^2(m{+}\frac{2}{\lambda}x)\sigma^2_m
\end{align}
Substituting back into \eqref{eq:cov_calc_cond} we have:
\begin{align}
\sigma_{x\vert \mathcal S}^2&=\sum\limits_m \text{sinc}^2(m{+}\frac{2}{\lambda}x)\sigma^2_m-\sum\limits_{i\in I} \text{sinc}^2(m(i){+}\frac{2}{\lambda}x)\sigma^2_{m(i)}\notag\\
&=\sum\limits_{m\notin I} \text{sinc}^2(m{+}\frac{2}{\lambda}x)\sigma^2_m
\end{align}
and this is similar to the optimization over the selection of the first location in \eqref{eq:sigma_11_opt} with the optimum achieved by selecting $x$ such that $\frac{2}{\lambda}x=m$ for the first $m\notin I$ which is the next adjacent location on the $\frac{\lambda}{2}$ grid which completes the proof.
\bibliographystyle{IEEEtran}
\bibliography{Bayesian_imaging_bib}

%
%

\end{document}